\documentclass[10pt,twoside]{article}
\usepackage[margin=1in]{geometry}

\usepackage[utf8]{inputenc} 
\usepackage[T1]{fontenc}    
\usepackage{url}            
\usepackage{booktabs}       
\usepackage{nicefrac}       
\usepackage{microtype}      
\usepackage{subfig}

\usepackage{amsfonts,amsmath,amssymb,amsthm,bm,bbm}

\usepackage{paralist}
\usepackage[boxed,lined]{algorithm2e}
\usepackage{epsfig}
\usepackage{wrapfig}
\usepackage{graphicx}
\usepackage{dsfont,array}

\usepackage{newclude}

\usepackage{xfrac}

\usepackage[inline]{enumitem}

\usepackage{multirow}
\usepackage{multicol}
\usepackage{xspace}
\usepackage{epstopdf}
\usepackage{mathtools}
\usepackage{booktabs}
\usepackage{placeins}
\usepackage{authblk}

\newtheorem{definition}{Definition}
\newtheorem{theorem}{Theorem}
\newtheorem{lemma}{Lemma}
\newtheorem{corollary}{Corollary}

\date{\today}
\begin{document}
\title{Coded Caching with Linear Subpacketization \\  is Possible 
using Ruzsa-Szem\'eredi Graphs}
%


\author[1]{Karthikeyan Shanmugam}
\author[2]{Antonia M. Tulino}
\author[3]{Alexandros G. Dimakis}
\affil[1]{IBM Research, T. J. Watson Research Center \protect \\
\texttt{karthikeyanshanmugam88@gmail.com}}
\affil[2]{Nokia Bell Labs, Holmdel NJ \&
Universit\'a Federico II,  Napoli Italy \protect \\
\texttt{a.tulino@nokia-bell-labs.com}}
\affil[3]{University of Texas at Austin \protect \\
\texttt{dimakis@austin.utexas.edu }}

\maketitle

\begin{abstract} 
	
Coded caching is a problem where encoded broadcasts are used to satisfy users requesting popular files and having caching capabilities. Recent work by Maddah-Ali and Niesen showed that it is possible to satisfy a scaling number of users with only a constant number of broadcast transmissions by exploiting coding and caching. 
Unfortunately, all previous schemes required the splitting of files into an exponential number of packets before the significant coding gains of caching appeared. The question of what can be achieved with polynomial subpacketization (in the number of users) has been a central open problem in this area. 
We resolve this problem and present the first coded caching scheme with polynomial (in fact, linear) subpacketization.
We obtain a number of transmissions that is not constant, but can be any polynomial in the number of users with an exponent arbitrarily close to zero. 
Our central technical tool is a novel connection between Ruzsa-Szem\'eredi graphs and coded caching. 

\end{abstract}

\section{Introduction}

The number and display capabilities of devices connected to the Internet through cellular wireless is rapidly increasing. 
This is leading to a dramatic increase in wireless traffic, predominantly driven by video content demand~\cite{femto1,paschos2016wireless}. System designers fortify 5G wireless networks with infrastructure densification and increased rates, as much as possible, but 
backhaul links become congested as the number of base stations and rates increase~\cite{bastug2014living}.
The idea of Femtocaching~\cite{golrezaei2012femtocaching,femto1} was to \textit{add storage to small-cell base stations} so that they can cache popular content and hence relieve the backhaul links. These links can even be wireless and pre-populate caches during low-traffic times, effectively converting the wireless infrastructure to a content delivery network (CDN) at the edge. 
Caching for wireless networks has been recognized as a disruptive technology for 5th generation (5G) cellular networks~\cite{boccardi2014five,wang2014cache}.

Caching helps in two fundamentally different ways. The first is the obvious one: If we cache popular content at a nearby small-cell 
a user can find it with local communication and cause less interference and no overheads to the macro-cell. 
Resource allocation for this problem is still highly non-trivial: Even for a known, fixed popularity distribution, deciding what to cache where is an NP-hard combinatorial optimization problem. Interestingly this problem can be provably approximated using submodular function optimization~\cite{femto1}. 
There are several other challenging aspects that must be addressed like including device-to-device communication~\cite{golrezaei2013femtocaching}, learning popularity distributions~\cite{blasco2014learning} and mobility~\cite{poularakis2014approximation}, among others. Resource allocation for caching in wireless networks is currently a rich area of active research. 

Caching also helps in a second, surprising way: even when the cached content is not desired nearby it may be desired far
away. This may create broadcast coding opportunities where a macro-cell transmits the XOR of two files 
and each user uses their cache to cancel the non-desired file. 
This caching benefit is unique to wireless channels and is an example of index coding~\cite{bar2011index}, a fundamental problem in information theory. The study of optimal design of caches to exploit this benefit due to coding was initiated by Maddah-Ali and Niesen in their pioneering work~\cite{AliNiesen,maddah2013decentralized}, often called the \textit{coded caching} problem.

In the \textit{Coded caching problem}~\cite{AliNiesen}, there is a library of $N$ files and $K$ users. Each user requests one of the files and has a local cache that can fit the equivalent of $M$ files. 
The problem is to design what each user caches from the library, before seeing what each user requests:
The \textit{placement phase} has to be designed such that, in the subsequent \textit{delivery phase}, \textit{any} set of user demands 
from the library can be satisfied using at most $R$ transmissions. In other words, all the users may decide to request the second file, or different files (out of the $N$ possible) and a base station has to decide what to transmit. Each user's cache now acts as pre-designed side information which can be used to broadcast linear combinations of packets. Coded caching can be seen as index coding with a twist:
the side information sets can be designed, but the user requests are selected subsequently, in an adversarial fashion. 

Since there are $K$ users, each requesting one file, the base station can trivially satisfy everyone using $R=K$ transmissions, even without any use of caching. Furthermore, every user can cache $\frac{M}{N}$ of every file. The base station can then transmit the remaining 
$1- \frac{M}{N}$ of each file, requiring in total  $R= K (1-\frac{M}{N})$ transmissions. 
This is just a conventional use of caching; slightly better than the trivial $R=K$ but not significantly, since it also scales linearly\footnote{Throughout this paper we assume (when we informally write about scaling laws) that $M/N$ is a constant, and $K$ scales to a large number of users. Despite this, our mathematical claims hold more generally, for any scaling relation between $M,N,K$ and we include all the parameters explicitly in our theorems. }
in the number of users $K$. 

The surprising result of Maddah-Ali and Niesen~\cite{AliNiesen} is that using coded caching it is possible to satisfy any request pattern using only $R = \frac{N}{M}$ transmissions, i.e. a \textit{constant} number of transmissions, not scaling in the number of users $K$. 
A number of subsequent works~\cite{ji2015order,ji2014average,niesen2013coded,karamchandani2016hierarchical} have introduced order optimal schemes for variations of the problem and established similar strong performance bounds. 

The central limitation of all previous works is the problem of \textit{exponential subpacketization}: Coded caching requires every file to be separated into $F$ subpackets that are then stored in user caches. 
In the original scheme~\cite{AliNiesen}, the number of packets $F$ was scaling exponentially in the number of users $K$ (for constant $M/N$).  
Maddah-Ali and Niesen~\cite{maddah2013decentralized} also showed that even if the placement is random, i.e. every user caches a random $\frac{M}{N}$ fraction of every file, then the same rate $R \approx \frac{N}{M}$ under adversarial user demands is possible. 
These schemes are called decentralized caching schemes since very little coordination is required among the user caches. 
This analysis required that $F \rightarrow \infty$. Subsequently, \cite{ShanmugamIT} showed that 
for any $F \leq  \exp( K \, M/N)$, a linear number of transmissions $R \geq c K$ would be required for some constant $c$. Therefore, the benefits of coded caching appear only when $F$ scales exponentially in the number of users $K$, making these schemes non-practical.

The problem of subpacketization is a central one in caching, since it was realized in~\cite{ShanmugamIT}.
Very recent work~\cite{tang2016coded,yan2015placement,shangguan2016centralized,yan2016placement}
has investigated better achievability schemes with a smaller number of packets, for the centralized case: 
Tang and Ramamoorthy~\cite{tang2016coded} connected coded caching to resolvable combinatorial designs and showed a scheme where $R \approx \frac{N}{M}$ while improving exponentially over the previous scheme by \cite{AliNiesen}, 
but still requiring exponential subpacketization. A similar result was shown in \cite{yan2015placement} through
a very interesting construction called a ``Placement Delivery Array'' (PDA). Shangguan et al.~\cite{shangguan2016centralized} connected this problem to the construction of some hypergraphs with extremal properties and showed the first sub-exponential (but still intractable) scheme: $F = \exp \left(\sqrt{K} \right)$.
Furthermore, they established that such hypergraphs do not exist when $F= O(K)$ for a constant number of transmissions $R$. 
Yan et al.~\cite{yan2016placement} connected the PDA formulation to another graph problem and derived a centralized caching scheme
with optimal non-scaling $R$, but $F$ scales super-polynomially in the number of users $K$.

\noindent \textbf{Our Contribution:} We present the first coded caching scheme with polynomial (in fact, linear) subpacketization. 
For any arbitrarily small constant $\delta$, we show that there is a small $\epsilon =\epsilon(\delta) >0$ such that we can design a centralized coded caching scheme that satisfies $K$ users with $R= K^{\delta}$ transmissions using only $F=K$ packets.  
This is possible for any constant ratio $\frac{M}{N}$, in fact it suffices if $\frac{M}{N} \geq K^{-\epsilon}$. 

Our central technical tool is a novel connection between Ruzsa-Szem\'eredi graphs~\cite{ruzsa1978triple,alon2012nearly} and coded caching. 
We show that any $(r,t)$ Ruzsa-Szem\'eredi graph on $K$ vertices and $\Theta(K^2)$ edges 
corresponds to a centralized coded caching scheme with $F=K$ and constant (or $o(1)$) $M/N$. Leveraging this, we show that an existing Ruzsa-Szem\'eredi construction from \cite{alon2012nearly} gives an explicit coded caching scheme with $R= K^{\delta}$, $\frac{M}{N} \geq K^{-\epsilon}$ and $F=K$ for any small $0<\delta<1$,  large enough $K$ and with $\epsilon = c_1\delta\exp ( -\frac{c_2}{ \delta})$ for some positive constants $c_1$ and $c_2$.
Note that our results do not violate the impossibility results established in~\cite{shangguan2016centralized}
since they apply for linear file sizes and constant $R,\, M/N$.

We emphasize that the graph constructions we use~\cite{alon2012nearly} require very large $K$ and do not seem directly applicable for real caching systems. In any case, we find it quite surprising that linear packet sizes $F=K$ are possible for near-optimal number of transmissions.
We believe this path can lead to practical coded caching schemes.

\section{Ruzsa-Szem\'eredi Graphs and Coded Caching }

We start with the definition of Ruzsa-Szem\'eredi Graphs and establish the connection to the centralized coded caching problem. 
Some graph theoretic notions are reviewed before that: Let $G(V,E)$ be an undirected graph on the set of vertices $V$ and edge set $E$. $G_{S}$ denotes the induced subgraph on the vertex set $S \subseteq V$. A \textit{matching} $M \subseteq E$ is a subset of edges which are pairwise disjoint (no two edges have a common vertex). 

\begin{definition}
The set of edges $M \subseteq E$ is an induced matching if for the set $S$ of all the vertices incident on the edges in $M$, the induced graph $G_S$ contains no other edges apart from those in the matching $M$. 
\end{definition}
\begin{definition}
 An $(r,t)$- Ruzsa-Szem\'eredi graph is a graph on $n$ vertices such that the edge set $E$ can be \textit{partitioned} into $t$ pairwise edge-disjoint induced-matchings of average size $r$.
\end{definition}
\textbf{Remark:} Note that the usual definition of Ruzsa-Szem\'eredi graph is when all the induced matchings have the exact same size $r$. Here, we adopt a weaker definition as it suffices to draw a connection to the coded caching problem.

 \begin{figure*}
 \centering
 \includegraphics[width=15cm]{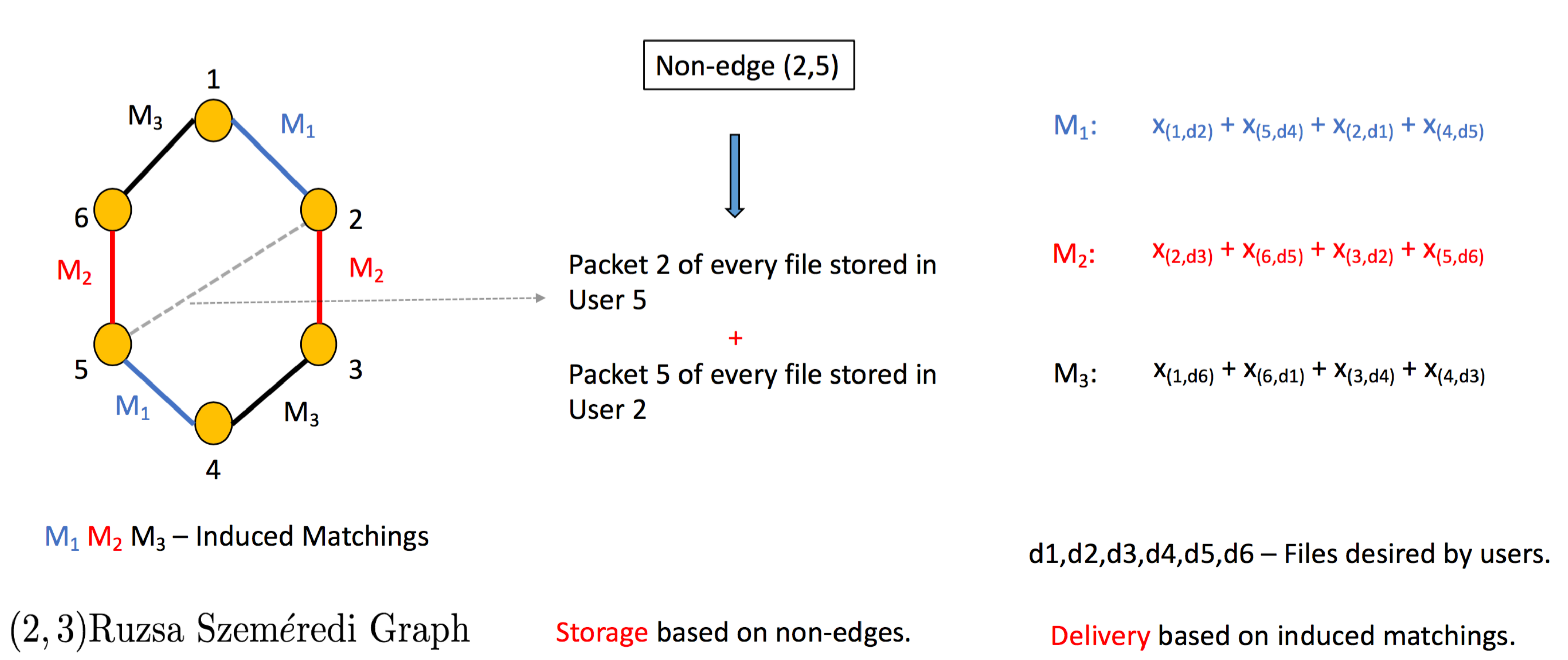}
 \caption{Illustration of the coded caching scheme described in Definition \ref{ruzsascheme} using a $(2,3)$ Ruzsa-Szem\'eredi graph on $6$ vertices. Every non-edge represents a storage action in the placement phase. Every induced matching corresponds to an XOR in the delivery phase. User $1$ recovers the packet $(2,d_1)$ since it has packets $1,5$ and $4$ of every file as $(1,1)$ (a self-loop non- edge), $(1,5)$ and $(1,4)$ are non-edges in the graph. }
 \label{fig:illustration}
 \end{figure*}

\hspace{1cm}

To review the coded caching problem: there is a noiseless broadcast channel between a server that has access to a library of $N$ files and $K$ users each of whom is equipped with a local cache of size $M$. The problem is to design the \textit{placement phase} where the user cache is filled by file packets from the library subject to the cache constraint and the \textit{delivery phase} where any demand pattern of the $K$ users arising from the library is satisfied with at most $R$ file transmissions.

Let $\mathbf{x}_i \in \mathbb{F}^{F \times 1}$ denote a vector of length $F$ over the field $\mathbb{F}$ presenting the $i$-th file in the library. Based on this, we define a $(R,K,M,N,F)$ coded caching scheme as follows:

\begin{definition}
  Let every file in the library contain $F$ packets. An (R,K,M,N,F) coded caching scheme consists of:
  \begin{enumerate}
   \item A family of subsets $\{S_{p,f}\}_{p \in [1:N], f \in [1:K]}$ where $S_{p,f} \subseteq [1:K]$ is the set of user caches where the packet $p$ of file $f$ is stored.  The placement is such that every user cache contains at most $MF$ file packets.
    \item For any set of user demands arising from the library of $N$ files, i.e. $\{ d_1,d_2 \ldots d_K\} \in [1:N]^{K}$, we define the transmission function $\phi \left(\mathbf{x}_{d_1}, \ldots \mathbf{x}_{d_K} \right) \rightarrow \mathbb{F}^{RF \times 1}$ for some field $\mathbb{F}$ such that every user can decode their demanded packets $\mathbf{x}_{d_k}$ using $\phi(\cdot)$ and the cache content available.
    \item The number of file transmissions is at most $R$ for any demand pattern among the users arising from the library.
    \end{enumerate}
  \end{definition}
  \textbf{Remark:} For some specific demand patterns, $RF$ transmissions may not be needed. However, we zero pad the transmissions to get $RF$ packets over $\mathbb{F}$ for notational convenience.
  
We can now establish the connection between Rusza-Szem\'eredi Graphs and Coded Caching Schemes. This is inspired by the Placement Delivery Array formulation of coded caching from~\cite{yan2015placement,shangguan2016centralized,yan2016placement}. Although there is a connection between square symmetric PDA and a Rusza-Szemeredi graph with appropriate parameters, we will state the connection to coded caching and prove it in a more direct fashion. 
 
We call the following coded caching scheme a Ruzsa-Szem\'eredi scheme for any $(r,t)$- Ruzsa-Szemeredi graph on $K$ vertices. 
  
\begin{definition}\label{ruzsascheme}
  (\textbf{Ruzsa-Szem\'eredi Coded Caching Scheme}) Let $G(V,E)$ be an $(r,t)$ Ruzsa-Szem\'eredi graph on $K$ vertices with $t$ pairwise edge disjoint induced matchings $M_1,M_2 \ldots M_t$ with average size 
 \[
 r = \frac{1}{t} \sum \limits_{i=1}^t \lvert M_i \rvert = \frac{1}{t} \lvert E \rvert,
 \] 
and minimum degree at least 
\[
(1 - \frac{M}{N}) \,K.
\] 
We now describe a placement and a feasible delivery scheme guided by the structure of the graph called the $(K,t,r, \frac{M}{N})$ Ruzsa-Szem\'eredi Scheme.
 \begin{enumerate}
  \item Let us consider an edge $\{i,j\} \notin E$. Every file in the coded caching scheme is split into $F=K$ packets. Then, we store the $i$-th packet of \textit{all} files in the library in the cache of user $j$ and we store the $j$-th packet of \textit{all} files in the library in the cache of user $i$.   
 \item  The delivery is based on the induced matchings in $G$. The edge-set of $G$ partitions into $t$ induced matchings of average size $r$. For any $1 \leq i \leq t$, consider the matching $M_i$ of size $r_i$. Let the matching consist of the $r_i$ edges 
\[
\left(p_1,p_2 \right), \left(p_3,p_4 \right) \ldots \left( p_{2r_i-1,2r_i} \right).
\]
We XOR the $2r_i$ packets:
\[
\{ (p_{2k-1},d_{p_{2k}}), (p_{2k},d_{p_{2k-1}}) \}_{1 \leq k \leq r_i}.
\]
The delivery scheme consists of $t$ packet transmissions and hence $t/F$ file transmissions.
  \end{enumerate}
  \end{definition}

We present an illustration of the coded caching scheme in the above definition through an example in Figure \ref{fig:illustration}. Figure \ref{fig:illustration} is based on a simple $(2,3)$ Ruzsa-Szem\'eredi graph on $6$ vertices. Every non-edge represents a storage action in the placement phase. Every induced matching corresponds to an XOR in the delivery phase.

 \begin{theorem}\label{thmcoded}
 Any $(K,t,r,\frac{M}{N})$ Rusza-Szemeredi Scheme is an $(R= \frac{t}{K},K,M,N,F=K)$ coded caching scheme.
 \end{theorem} 
 \begin{proof}
 The number of users $K$ is equal to the number of packets per file $F$ in the construction, i.e. $K=F$. We first check if the memory constraint is satisfied. For a specific user $j$, the number of packets per file stored at the user is $K-d_j$ where $d_j$ is the degree of vertex $j$. Therefore, the total number of file packets stored is
 \[
  (K-d_i) N \leq \frac{M}{N}K N = MK = MF. \]
  
Therefore, the cache memory constraint of every user is satisfied. Further, the number of file transmissions ($R$) is seen to be $t/F$.

We now show that the scheme above satisfies any set of user demands. Consider the matching $M_i$ consisting of edges $r_i$ edges 
\[ \left(p_1,p_2 \right), \left(p_3,p_4 \right) \ldots \left( p_{2r_i-1,2r_i} \right), \]
and the XOR transmission corresponding to this matching consisting of $2r_i$ packets:
\[ \{ (p_{2k-1},d_{p_{2k}}), (p_{2k},d_{p_{2k-1}}) \}_{1 \leq k \leq r_i}. \] 

First, we show that, from the XOR consisting of the packets mentioned above, every user in 
 $\{p_1,p_2 \ldots p_{2r_i} \}$ 
recovers at least one packet used in the XOR that they desired. 
This decoding is done at every user by just using this XOR and the cache content. It is enough to show that every user $ k \in \{ p_1,p_2 \ldots p_{2r_i} \}$,
has all but $1$ packet from the $2r_i$ packets used in the XOR in its cache. We consider the following two cases:
  
  \begin{enumerate}
   \item User $p_{q}: q \in \{ 1,2 \ldots 2r_i \} $, $q$ is odd: Then, the user desires the packet $(p_{q+1},d_{p_q})$ from the XOR formed from the matching $M_i$. Suppose, user $p_q$ does not possess packet $p_m$ of some file for $m \neq q+1,~ m \in \{1, \ldots 2r_i \}$. Then, it means that $\{p_m,p_q\}$ is an edge in $E$ as according to the placement described, as only non edges correspond to placement actions. This is impossible since vertices $p_m,p_q$ belong to the induced matching $M_i$ and there is no edge between $p_m$ and $p_q$ for $m \neq q+1$ (Note that there is no self-loop in the graph). Hence, user $p_q$ possesses packet $p_m$ belonging to all files where $m \neq q+1,~m \in \{1,2 \ldots 2r_i\}$.
   
   \item User $p_{q}: q \in \{ 1,2 \ldots 2r_i \} $, $q$ is even. A similar argument holds but we will give the details for the sake of completeness. The user desires the packet $(p_{q-1},d_{p_q})$ from the XOR formed from matching $M_i$. Suppose, user $p_q$ does not possess packet $p_m$ of some file for $m \neq q-1,~ m \in \{1, \ldots 2r_i \}$. Then, it means that $\{p_m,p_q\}$ is an edge in $E$ as according to the placement described, as only non edges correspond to placement actions. This is impossible since vertices $p_m,p_q$ belong to the induced matching $M_i$ and there is no edge between $p_m$ and $p_q$ for $m \neq q-1$. Hence, user $p_q$ possesses packet $p_m$ belonging to all files where $m \neq q-1,~ m \in \{1,2 \ldots 2r_i\}$.
  \end{enumerate}
   
   Now, we argue that all demanded packets are delivered to their respective users using the matching based delivery scheme. Consider a packet $(f,d_k),~1 \leq f \leq K=F,~1 \leq k \leq K$ demanded by user $k$. Suppose this packet was not cached by user $k$, then $(f,k)$ is an edge in the graph and because the edge set partitions into induced matchings, it belongs to some induced matching $M_q$, $1 \leq q \leq t$. Therefore, the $q$-th XOR in the above delivery scheme would enable user $k$ to decode the packet $(f,d_k)$. This completes the proof.
    
 \end{proof}
  
  In what follows, we consider a known construction of a dense $(r,t)$-Ruzsa-Szem\'eredi graph that decomposes into almost near linear number of induced matchings and indicate how it gives rise to coded caching schemes when the file size is linear in the number of users using the connection described above.
  
  \subsection{The Ruzsa-Szem\'eredi Construction by Alon, Moitra and Sudakov}
    In this section, we describe a coded caching scheme with almost constant number of transmissions ($R$ is almost a constant) and file size $F$ required is linear in $K$ even for a sub-constant user cache to library size ratio ($\frac{M}{N}$) We describe a Ruzsa-Szem\'eredi graph from \cite{alon2012nearly}. 
    
  \begin{definition}\label{Ruzsa-Alon}
   Define a graph $G(V,E)$ as follows:
    $V=[C]^n$ where $C$ is a positive integer. $|V|=K=C^n$. $n$ is even, $n \geq 2C$. Suppose $x$ and $y$ are sampled uniformly from $V$, let $\mu = \mathbb{E}_{x,y} \left[ \lVert x-y \rVert_2^2\right]$. Consider a pair of ordered tuples $u,v \in [C]^n$. Now, $(u,v) \in E$ if and only if $ \lVert \lVert u-v \rVert_2^2 - \mu \rVert< n$.
    \end{definition}
      
  \begin{theorem}\label{thmalon}
      \cite{alon2012nearly} The graph $G$ defined in Definition \ref{Ruzsa-Alon} is an $(r,t)$-Ruzsa-Szem\'eredi graph on $K=C^n$ vertices where the edge set partitions into $t=K^{f}$ edge disjoint induced matchings and it has at most $K^g$ edges missing such that 
\[ 
 g=2-\frac{1}{2C^4 \mathrm{ln} C}+o(1), 
 \]
 and 
 \[ f=1+ 2 \frac{\mathrm{ln} 10.5}{\mathrm{ln} C}+o(1).
 \]
  \end{theorem}
  
  We make the following simple observation to make it easier to connect it to the coded caching scheme.
    \begin{lemma}\label{lem:degree}
     Consider any vertex $x \in V=[C]^n$ in $G$ as in Definition \ref{Ruzsa-Alon}. Let $K=C^n$. The degree of $x$ is:
      \[ \lvert \{y \in V: (x,y) \in E \} \rvert \geq K \left(1-2K^{- \frac{1}{2C^4 \mathrm{ln} C}} \right). \]
    \end{lemma}
    \begin{proof}
       The proof is identical to Claim $2.1$ in \cite{alon2012nearly} except that this is for neighbors of a single vertex. Consider $y$ sampled uniformly randomly from $V$. Then $(x_i -y_i)^2$ is a bounded random variable with support $[0,C^2]$ and i.i.d with respect to the index $i$. Therefore applying Hoeffding's inequality, we have:
          \begin{align}\label{eqn:degreebnd}
           \lvert \{y \in V: (x,y) \notin E \} \rvert & \leq K *\mathrm{Pr}_{y} \left( \lvert \lVert x-y \rVert_2^2 -\mu  \rvert > n \right) \nonumber \\
            & \leq 2K\exp \left( - \frac{n}{2C^4} \right) = 2 K^{1- \frac{1}{2C^4 \mathrm{ln} C}}. 
           \end{align}
       The degree of vertex $x$ is $K - \lvert \{y \in V: (x,y) \notin E \} \rvert$. Therefore, (\ref{eqn:degreebnd}) implies the result.  
    \end{proof}
   
     \begin{theorem}
        With the construction of \cite{alon2012nearly} given in Definition \ref{Ruzsa-Alon}, the corresponding Ruzsa-Szem\'eredi scheme according to Definition \ref{ruzsascheme} yields a
    $(R= K^{ 2 \frac{\mathrm{ln} 10.5}{\mathrm{ln} C}+o(1)},K,M,N,F=K)$ coded caching scheme whenever $ \frac{M}{N} \geq 2K^{- \frac{1}{2C^4 \mathrm{ln} C}} $ where $K = C^n$, $n$ even and $ n \geq 2C$.  
     \end{theorem}
     \begin{proof}
       Theorem \ref{thmalon} shows that the Ruzsa-Szem\'eredi graph construction in Definition \ref{Ruzsa-Alon} has the following properties: 
      \begin{enumerate} 
       \item It is on $K$ vertices such that $K=C^n, n \geq 2C$ 
        \item  Number of edge disjoint induced matchings is $t= K^{1+ 2 \frac{\mathrm{ln} 10.5}{\mathrm{ln} C}+o(1)}$ 
        \item  Average size of each matching is at least $r=K^{2}\left(1-K^{-\frac{1}{2C^4 \mathrm{ln} C}+o(1)} \right)/t$. 
       \end{enumerate} 
        Further, Lemma \ref{lem:degree} guarantees that the minimum degree of this construction is at least $K (1- 2K^{- \frac{1}{2C^4 \mathrm{ln} C}})$. By Definition \ref{ruzsascheme}, this yields a $\left(K, t= K^{1+ 2 \frac{\mathrm{ln} 10.5}{\mathrm{ln} C}+o(1)},r=K^{2}\left(1-K^{-\frac{1}{2C^4 \mathrm{ln} C}+o(1)} \right)/t ,\frac{M}{N}\right)$ Ruzsa-Szem\'eredi scheme whenever $\frac{M}{N} \geq 2K^{- \frac{1}{2C^4 \mathrm{ln} C}} $ . By Theorem \ref{thmcoded}, this implies a feasible coded caching scheme with parameters stated in the theorem. This completes the proof.
     \end{proof}
    
   We restate the result in a different way in the following corollary: 
      \begin{corollary}
     For any $0< \delta < 1$, the construction of \cite{alon2012nearly} given in Definition \ref{Ruzsa-Alon} gives a $(R= K^{ \delta}, K, M, N, F=K)$ coded caching scheme for $K \geq K \left( \delta \right)$ (large enough $K$)  with $ \frac{M}{N}  \geq 2K^{- c_1\delta \exp\left( -\frac{c_2}{\delta}\right)}$ where $c_1,c'_2$ are universal positive constants.
   \end{corollary}
   \begin{proof}
    We set $\delta= 2 \frac{\mathrm{ln} 10.5}{\mathrm{ln} C}$. Then, $\frac{1}{2C^4 \mathrm{ln} C} =c_1\delta \exp \left( - \frac{c_2}{\delta} \right)$ for some universal positive constants $c_1,c_2$. For large enough $C$, $\delta$ can be made to lie in $(0,1)$ and correspondingly $K$ increases as function of $\delta$ since $K=C^n$ by construction. Substituting the re-parameterized values proves the result.
   \end{proof}
   \textbf{Remark:} This shows that, for sub-constant $\frac{M}{N}$ and large enough $K$, there are coded caching schemes with file size $F$ linear in $K$ with almost constant number of transmissions.
   
  \section{Conclusion}
     We presented a connection between Ruzsa-Szem\'eredi Graphs studied in extremal combinatorics and the coded caching problem. We showed that a dense Ruzsa-Szem\'eredi construction gives coded caching schemes with file size linear in the number of users. We leveraged an existing graph construction to show that coded caching schemes with file size linear in the number of users, sub-constant user cache to library size ratio and almost constant transmission rate for any set of demands is possible when the number of users is large. This settles an outstanding question of existence of schemes with linear file size and very small rate in this line of work. The construction used in this work requires the number of users $K$ to be very large which may be impractical. It would be of great interest to construct Ruzsa-Szem\'eredi graphs (more general bipartite versions) in the dense regime where the number of vertices is not very large. 
 \bibliographystyle{plain}
\bibliography{refruzsa}

\begin{thebibliography}{10}

\bibitem{alon2012nearly}
Noga Alon, Ankur Moitra, and Benny Sudakov.
\newblock Nearly complete graphs decomposable into large induced matchings and
  their applications.
\newblock In {\em Proceedings of the forty-fourth annual ACM symposium on
  Theory of computing}, pages 1079--1090. ACM, 2012.

\bibitem{bar2011index}
Ziv Bar-Yossef, Yitzhak Birk, TS~Jayram, and Tomer Kol.
\newblock Index coding with side information.
\newblock {\em IEEE Transactions on Information Theory}, 57(3):1479--1494,
  2011.

\bibitem{bastug2014living}
Ejder Bastug, Mehdi Bennis, and M{\'e}rouane Debbah.
\newblock Living on the edge: The role of proactive caching in 5g wireless
  networks.
\newblock {\em IEEE Communications Magazine}, 52(8):82--89, 2014.

\bibitem{blasco2014learning}
Pol Blasco and Deniz Gunduz.
\newblock Learning-based optimization of cache content in a small cell base
  station.
\newblock In {\em Communications (ICC), 2014 IEEE International Conference on},
  pages 1897--1903. IEEE, 2014.

\bibitem{boccardi2014five}
Federico Boccardi, Robert~W Heath, Angel Lozano, Thomas~L Marzetta, and Petar
  Popovski.
\newblock Five disruptive technology directions for 5g.
\newblock {\em IEEE Communications Magazine}, 52(2):74--80, 2014.

\bibitem{golrezaei2013femtocaching}
Negin Golrezaei, Andreas~F Molisch, Alexandros~G Dimakis, and Giuseppe Caire.
\newblock Femtocaching and device-to-device collaboration: A new architecture
  for wireless video distribution.
\newblock {\em IEEE Communications Magazine}, 51(4):142--149, 2013.

\bibitem{golrezaei2012femtocaching}
Negin Golrezaei, Karthikeyan Shanmugam, Alexandros~G Dimakis, Andreas~F
  Molisch, and Giuseppe Caire.
\newblock Femtocaching: Wireless video content delivery through distributed
  caching helpers.
\newblock In {\em INFOCOM, 2012 Proceedings IEEE}, pages 1107--1115. IEEE,
  2012.

\bibitem{ji2014average}
Mingyue Ji, Antonia~M Tulino, Jaime Llorca, and Giuseppe Caire.
\newblock On the average performance of caching and coded multicasting with
  random demands.
\newblock In {\em 2014 11th International Symposium on Wireless Communications
  Systems (ISWCS)}, pages 922--926. IEEE, 2014.

\bibitem{ji2015order}
Mingyue Ji, Antonia~M Tulino, Jaime Llorca, and Giuseppe Caire.
\newblock Order-optimal rate of caching and coded multicasting with random
  demands.
\newblock {\em arXiv preprint arXiv:1502.03124}, 2015.

\bibitem{karamchandani2016hierarchical}
Nikhil Karamchandani, Urs Niesen, Mohammad~Ali Maddah-Ali, and Suhas~N Diggavi.
\newblock Hierarchical coded caching.
\newblock {\em IEEE Transactions on Information Theory}, 62(6):3212--3229,
  2016.

\bibitem{AliNiesen}
M.~A. Maddah-Ali and U.~Niesen.
\newblock Fundamental limits of caching.
\newblock {\em IEEE Transactions on Information Theory}, 60(5):2856--2867, May
  2014.

\bibitem{maddah2013decentralized}
Mohammad~Ali Maddah-Ali and Urs Niesen.
\newblock Decentralized coded caching attains order-optimal memory-rate
  tradeoff.
\newblock {\em arXiv preprint arXiv:1301.5848}, 2013.

\bibitem{niesen2013coded}
Urs Niesen and Mohammad~Ali Maddah-Ali.
\newblock Coded caching with nonuniform demands.
\newblock {\em arXiv preprint arXiv:1308.0178}, 2013.

\bibitem{paschos2016wireless}
Georgios Paschos, Ejder Bastug, Ingmar Land, Giuseppe Caire, and M{\'e}rouane
  Debbah.
\newblock Wireless caching: Technical misconceptions and business barriers.
\newblock {\em IEEE Communications Magazine}, 54(8):16--22, 2016.

\bibitem{poularakis2014approximation}
Konstantinos Poularakis, George Iosifidis, and Leandros Tassiulas.
\newblock Approximation algorithms for mobile data caching in small cell
  networks.
\newblock {\em IEEE Transactions on Communications}, 62(10):3665--3677, 2014.

\bibitem{ruzsa1978triple}
Imre~Z Ruzsa and Endre Szemer{\'e}di.
\newblock Triple systems with no six points carrying three triangles.
\newblock {\em Combinatorics (Keszthely, 1976), Coll. Math. Soc. J. Bolyai},
  18:939--945, 1978.

\bibitem{shangguan2016centralized}
Chong Shangguan, Yiwei Zhang, and Gennian Ge.
\newblock Centralized coded caching schemes: A hypergraph theoretical approach.
\newblock {\em arXiv preprint arXiv:1608.03989}, 2016.

\bibitem{femto1}
K.~Shanmugam, N.~Golrezaei, AG. Dimakis, AF. Molisch, and G.~Caire.
\newblock Femtocaching: Wireless content delivery through distributed caching
  helpers.
\newblock {\em Information Theory, IEEE Transactions on}, 59(12):8402--8413,
  Dec 2013.

\bibitem{ShanmugamIT}
K.~Shanmugam, M.~Ji, A.~M. Tulino, J.~Llorca, and A.~G. Dimakis.
\newblock Finite-length analysis of caching-aided coded multicasting.
\newblock {\em IEEE Transactions on Information Theory}, 62(10):5524--5537, Oct
  2016.

\bibitem{tang2016coded}
Li~Tang and Aditya Ramamoorthy.
\newblock Coded caching with low subpacketization levels.
\newblock {\em arXiv preprint arXiv:1607.07920}, 2016.

\bibitem{wang2014cache}
Xiaofei Wang, Min Chen, Tarik Taleb, Adlen Ksentini, and Victor Leung.
\newblock Cache in the air: exploiting content caching and delivery techniques
  for 5g systems.
\newblock {\em IEEE Communications Magazine}, 52(2):131--139, 2014.

\bibitem{yan2015placement}
Qifa Yan, Minquan Cheng, Xiaohu Tang, and Qingchun Chen.
\newblock On the placement delivery array design in centralized coded caching
  scheme.
\newblock {\em arXiv preprint arXiv:1510.05064}, 2015.

\bibitem{yan2016placement}
Qifa Yan, Xiaohu Tang, Qingchun Chen, and Minquan Cheng.
\newblock Placement delivery array design through strong edge coloring of
  bipartite graphs.
\newblock {\em arXiv preprint arXiv:1609.02985}, 2016.

\end{thebibliography}
\end{document}